\documentclass[11pt,english]{article}
\usepackage[T1]{fontenc}
\usepackage{float}
\usepackage{amsthm}
\usepackage{amsmath}
\usepackage{amssymb}
\usepackage{setspace}
\usepackage{tikz}
\usepackage[latin9]{inputenc}
\usepackage{amsfonts}
\usepackage{amsmath}
\usepackage{amssymb}
\usepackage{amsthm}
\usepackage{graphicx}
\usepackage[colorlinks=true,linkcolor=blue,citecolor=red]{hyperref}
\usepackage[]{algorithm2e}
\graphicspath{}
\usepackage[pass,letterpaper,margin=1in]{geometry}
\usepackage{float}

\newcommand{\set}[1]{\{#1\}}

\newcommand{\abs}[1]{{\left|#1\right|}}
\newcommand{\zo}{{\set{0,1}}}
\newcommand{\pmk}{{\set{-k,\ldots,k}}}

\newcommand{\slide}{\text{slide}}

\newcommand{\Update}{\text{\it Update}}

\newcommand{\matureList}{\text{\it matureList}}
\newcommand{\Null}{\text{null}}
\newcommand{\rightMature}{\text{\it rightMature}}
\newcommand{\maturePeriod}{\text{\it maturePeriod}}

\newcommand{\startSlide}{\text{\it startSlide}}
\newcommand{\moveToMatures}{\text{\it moveToMatureDiagonals}}
\newcommand{\handleMatures}{\text{\it processMatureDiagonals}}
\newcommand{\newRightMature}{\text{\it resetRightMature}}

\newtheorem{theorem}{Theorem}[section]
\newtheorem{corollary}[theorem]{Corollary}
\newtheorem{proposition}[theorem]{Proposition}

\newtheorem{lemma}[theorem]{Lemma}
\newtheorem{claim}[theorem]{Claim}

\newenvironment{proof-sketch}{\noindent{\bf Sketch of Proof}\hspace*{1em}}{\qed\bigskip}
\newenvironment{proof-idea}{\noindent{\bf Proof Idea}\hspace*{1em}}{\qed\bigskip}
\newenvironment{proof-attempt}{\noindent{\bf Proof Attempt}\hspace*{1em}}{\qed\bigskip}

\floatstyle{ruled}
\newfloat{algorithm}{tbp}{loa}
\providecommand{\algorithmname}{Algorithm}
\floatname{algorithm}{\protect\algorithmname}

\floatstyle{ruled}
\newfloat{algorithm}{tbp}{loa}
\providecommand{\algorithmname}{Algorithm}
\floatname{algorithm}{\protect\algorithmname}

\theoremstyle{plain}

  \theoremstyle{definition}
  
  \theoremstyle{plain}
  
  \theoremstyle{plain}
  
  \theoremstyle{remark}

\makeatother

\usepackage{babel}
  \providecommand{\corollaryname}{Corollary}
  \providecommand{\definitionname}{Definition}
  \providecommand{\lemmaname}{Lemma}
\providecommand{\theoremname}{Theorem}

\usepackage{fullpage}
\SetKwInOut{Input}{Input}
\SetKwInOut{Output}{Output}
\begin{document}

	\thispagestyle{empty}
	
	\newcommand*\samethanks[1][\value{footnote}]{\footnotemark[#1}
	\title{Streaming Algorithms For Computing Edit Distance Without Exploiting Suffix Trees\thanks{The research leading to these results has received funding from the European Research Council under the European Union's Seventh Framework Programme (FP/2007-2013)/ERC Grant Agreement n. 616787. The first author was partially supported by ACM-India/IARCS Travel Grant and Research-I Foundation.}}
	
	\thispagestyle{empty}

	\author{
		Diptarka Chakraborty\\
		Department of Computer Science \& Engineering\\
		Indian Institute of Technology Kanpur\\
		Kanpur, India\\
		diptarka@cse.iitk.ac.in
		\and
		Elazar Goldenberg\\
		Charles University in Prague\\
		Computer Science Institute of Charles University\\
		Prague, Czech Republic\\
		elazargold@gmail.com.
		\and
		Michal Kouck{\' y}\\
		Charles University in Prague\\
		Computer Science Institute of Charles University\\
		Prague, Czech Republic\\
		koucky@iuuk.mff.cuni.cz 
	}

	\date{}
	\maketitle
	\begin{abstract}
		The edit distance is a way of quantifying how similar two strings are
		to one another by counting the minimum number of character insertions,
		deletions, and substitutions required to transform one string into the
		other.
		
		In this paper we study the computational problem of computing the edit
		distance between a pair of strings where their distance is bounded by a
		parameter $k\ll n$. We present two streaming algorithms for computing edit
		distance: One runs in time $O(n+k^2)$ and the other $n+O(k^3)$. By writing
		$n+O(k^3)$ we want to emphasize that the number of operations per an input
		symbol is a small constant. In particular, the running time does not
		depend on the alphabet size, and the algorithm should be easy to
		implement.
		
		Previously a streaming algorithm with running time $O(n+k^4)$ was given in
		the paper by the current authors (STOC'16). The best off-line algorithm
		runs in time $O(n+k^2)$ (Landau et al., 1998) which is known to be optimal
		under the Strong Exponential Time Hypothesis.
	\end{abstract}
	
	\section{Introduction}
	The \emph{edit distance} (aka \emph{Levenshtein distance})~\cite{Lev65} is a widely used distance measure between pairs of strings $x,y $ over some alphabet $\Sigma$. It finds applications in several fields like computational biology, pattern recognition, text processing, information retrieval and many more. The edit distance between $x$ and $y$, denoted by  $\Delta(x,y)$, is defined as the minimum number of character insertions, deletions, and substitutions needed for converting $x$ into $y$. Due to its immense applicability, the computational problem of computing the edit distance between two given strings $x$ and $y \in \Sigma^n$ is of prime interest to researchers in various domains of computer science. Sometimes one also requires that the algorithm finds an \emph{alignment} of $x$ and $y$, i.e., a series of edit operations that transform $x$ into $y$.

	In this paper we study the problem of computing edit distance of strings when given an {\em a priori} upper bound $k\ll n$ on their distance. 
This is akin to {\em fixed parameter tractability}.  Arguably, the case when the edit distance is small relative to the length of the strings
is the most interesting as when comparing two strings with respect to their edit distance we are implicitly making an assumption that the strings
are similar. If they are not similar the edit distance is uninformative. There are few exceptions to this rule, most notably the reduction
of instances of formula satisfiability (SAT) to instances of edit distance of exponentially large strings \cite{BI15} where the edit
distance of resulting strings is close to their length. However, such instance of the edit distance problem are rather artificial.
For typical applications the edit distance of the two strings is much smaller then the length of the strings. Consider for example copying DNA
during cell division: Human DNA is essentially a string of about $10^9$ letters from $\{A,C,G,T\}$, and due to imperfections in the copying
mechanism one can expect about 50 edit operations to occur during the process. So in many applications we can be looking for a handful
of edit operations in large strings.

        Landau et al.~\cite{LMS98} provided an algorithm
	that runs in time $O(n+k^2)$ and uses space $O(n)$ when size of the alphabet $\Sigma$ is constant. In general the running time of the algorithm given in~\cite{LMS98} is $O(n\cdot \min\{\log n, \log |\Sigma|\}+k^2)$. 
In this paper we revisit this question and study streaming algorithms for edit distance, that is, algorithms that make only one or few passes over the input $x$ and $y$.
We consider so called {\em synchronous} streaming model where $x$ and $y$ are processed left-to-right in parallel at about the same rate, and the internal
memory of the algorithm is limited compared to the size of the whole input.
We provide two algorithms for computing edit distance that run in a streaming fashion. One of them essentially matches the parameters of
the algorithm given by~\cite{LMS98}, improving on them slightly, while working in streaming fashion using only $O(k)$ internal memory.
The other one which we consider to be the main contribution of this paper differs slightly in its parameters but we believe it is superior in practicality.
	
	The algorithm given by~\cite{LMS98} relies on a suffix tree machinery and builds suffix trees for the entire input. 
 	While from theoretical perspective the task 
	of building a suffix tree requires only linear time for a constant-size alphabet, practically
	they are quite expensive to build because of hidden constants.  Moreover, for arbitrary-size alphabets suffix trees incur super-linear cost. More specifically, the known algorithms used to build a suffix tree of a string of length $n$ over alphabet $\Sigma$ run in time $O(n \cdot \min\{\log n, \log |\Sigma|\})$ \cite{Weiner73, Mc76,Ukkonen95}.

	Hence, for practical purposes people prefer the algorithm by~\cite{UKK85} to compute edit distance, despite its running time being $O(nk)$ cf.~\cite{PP08}.  The algorithm by~\cite{UKK85} does not build suffix trees. We propose a new approach for computing edit distance, which does not involve computing suffix trees either, yet, it improves over the running time of~\cite{UKK85} algorithm. 
	We obtain an algorithm that makes one-pass over its input $x$ and $y$, uses space $O(k)$ to compute the edit distance of $x$ and $y$
	(space $O(k^2)$ to compute an optimal alignment of $x$ and $y$) and runs in time $n+O(k^3)$. By writing $n+O(k^3)$ we want
	to emphasize that the number of operations per an input symbol is a small constant. (The constant in the big-$O$ is also reasonable.) Moreover, we emphasize that running time of our algorithm is independent of the alphabet size and thus to the best of our knowledge this is the first algorithm to compute edit distance that runs in ``truly'' linear time for small values of $k$. In that regard it is an improvement over the algorithm given in~\cite{LMS98} for large alphabets. We believe that due to its simplicity it should be relevant for practice. Formally our result is as follows:
	
	\begin{theorem}\label{thm:main}
		
		There is a deterministic algorithm that on input $x,y\in \Sigma^n$ and an integer $k$, such that $\Delta(x,y)\le k$ outputs $\Delta(x,y)$. The algorithm accesses $x$ and $y$ in one-way manner, runs in $c ( n+ k^3)$ time while using $O(k)$ space, where $c$ is a small constant. Moreover, one can output an optimal alignment between $x$ and $y$ while using extra $O(k^2)$ space. The algorithm never runs in time more than $O(kn)$.
		
	\end{theorem}
	
	The algorithm from the above theorem is efficient with respect to the memory access pattern when $x$ and $y$ are stored in the main memory. In the cache oblivious model with memory block size $B$ the algorithm performs only $O(\frac{n+k^3}{B})$ IO operations.
	
	Our second result confirms that from theoretical point of view, the running time of streaming algorithms is as good as that of the~\cite{LMS98} algorithm. We even improve slightly the dependency on the alphabet size. In particular, instead of $O(n\cdot \min\{\log n, \log |\Sigma|\}+k^2)$ running time of~\cite{LMS98}, we achieve $O(n\cdot \min\{\log k, \log |\Sigma|\}+k^2)$ running time. The algorithm uses only $O(k)$ space to compute the edit distance between the input strings and further $O(k^2)$ space for finding an optimal alignment.
%
	\begin{theorem}
		\label{thm:computeEdit}
		There is a deterministic algorithm that on input $x,y\in \Sigma^n$ and an integer $k$, such that $\Delta(x,y)\le k$ outputs $\Delta(x,y)$. The algorithm accesses $x$ and $y$ in one-way manner, runs in $O(n\cdot \min\{\log k, \log |\Sigma|\}+k^2)$ time while using $O(k)$ space. Moreover, one can output an optimal alignment between $x$ and $y$ while using extra $O(k^2)$ space. 
	\end{theorem}
	Previously, the best known (and only) streaming algorithm for edit distance was given in~\cite{CGK16}. That algorithm has running time $O(n+k^{1/4})$, uses space $O(k^4)$,
and is substantially more complex.

   All our algorithms output $\infty$ when run on strings which have edit distance larger than the parameter $k$. Hence, similarly to~\cite{CGK16},
if we allow the algorithms $O(\log \log k)$ passes over the input, we do not have to provide them the parameter $k$. One can search for
the smallest $k$ of the form $2^{{(1+\epsilon)}^i}$ for which the algorithm from Theorem~\ref{thm:main} returns a finite edit distance to obtain an algorithm that handles
all pairs of strings with running time $O(n \log \log n + k'^{3+3\epsilon})$, where $k'=\Delta(x,y)$.

	\subsection{Previous work}
	One can easily solve the problem of computing exact edit distance (the {\em decision} problem) in $O(n^2)$ time using a basic dynamic programming approach~\cite{WF74}. This bound was later slightly improved by Masek and Paterson~\cite{MP80} and they achieved an $O(n^2/\log n)$ time algorithm for finite alphabets. So far this is the best known upper bound for this problem. Recently, it was shown that this bound cannot be improved significantly unless the Strong Exponential Time Hypothesis is false~\cite{BI15, BK15}. They establishe this fact by providing a reduction which (implicitly) maps instances of SAT into instances of edit distance with the edit distance close to $n$. However, their result does not exclude the possibility of getting faster algorithms in the case of small edit distance.
	
	Suppose we are guaranteed that the edit distance between the two input strings is bounded by $k\ll n$. Then there are algorithms that are much more efficient in terms of both time and space. Ukkonen~\cite{UKK85} gave an algorithm to solve the decision problem in time $O(kn)$ and space $O(k)$. The same algorithm uses $O(n)$ space to find the optimal alignment (the {\em search} problem). Later, Landau et al.~\cite{LMS98} solved the decision problem within time $O(n\cdot \min\{\log n, \log |\Sigma|\}+k^2)$ and $O(n)$ space. By slightly modifying their algorithm the search problem can be solved as well using only $O(k^2)$ extra space. Interested readers may read a survey by Navarro~\cite{Nav01} for a comprehensive treatment on this topic. In a very recent development, Chakraborty, Goldenberg and Kouck{\' y}~\cite{CGK16} considered the search problem under the promise that the edit distance is small. They gave a single-pass algorithm that runs in time $O(n+k^4\cdot \min\{\log k, \log |\Sigma|\})$ while using space of size $O(k^4)$. The authors also mentioned that they can remove the promise by paying a penalty in the number of passes over the input and slightly worse time and space complexity. Table~\ref{TaxomonyTable} summarizes the above results.

	\begin{table}[ht]
		\centering
		\caption{Taxonomy of Algorithms Computing Edit Distance}
		\label{TaxomonyTable}
		\begin{tabular}{|l|l|l|l|l|}
			\hline
			Authors        & Time              & Space       \\ \hline
			\cite{WF74}   & $O(n^2)$          & $O(n)$                                                 \\ \hline
			\cite{MP80}    & \begin{tabular}[c]{@{}l@{}}$O(n^2/\log n)$\\ (for finite alphabets)\end{tabular}  & $O(n)$                                                  \\ \hline
			\cite{LMS98}   & $O(n\cdot \min\{\log n, \log |\Sigma|\}+k^2)$        & $O(n)$                                                 \\ \hline
			\cite{UKK85}    & $O(nk)$           & $O(k)$                                          \\ \hline		
			\cite{CGK16}     & \begin{tabular}[c]{@{}l@{}}$O(n+k^4\cdot \min\{\log k, \log |\Sigma|\})$\\ (randomized streaming and\\ single pass)\end{tabular} & $O(k^4)$                   								\\ \hline
			This paper     & \begin{tabular}[c]{@{}l@{}}$O(n\cdot \min\{\log k, \log |\Sigma|\}+k^2)$\\(streaming and single pass) \end{tabular} & $O(k)$                  								\\ \hline
			This paper     & \begin{tabular}[c]{@{}l@{}}$O(n+k^3)$\\(streaming and single pass) \end{tabular} & $O(k)$                    \\ \hline
		\end{tabular}
	\end{table}
	
	The problem of computing edit distance in the streaming model has been studied first time in~\cite{CGK16}. 
Independently of the current paper, Belazzougui and Zhang~\cite{BZ16} give an algorithm similar to our $O(n+k^2)$ time algoritm.
A related problem, namely \emph{edit distance to monotonicity}, which is equivalent to the problem of finding \emph{longest increasing subsequence}, has been studied extensively in streaming model~\cite{LVZ05, SW07, GJKK07, GG07, EJ08, CLLPTZ11, SS13}. However, the main focus of all of these results was to determine the upper and lower bound on the space requirement instead of time for exact as well as approximate solutions. Another important point to note is that the problem of edit distance to monotonicity is in some sense much easier because it can be solved in $O(n \log n)$ time~\cite{Sch61, Fred75} whereas general edit distance cannot be computed in strictly sub-quadratic time unless SETH is false~\cite{BI15, BK15}.
	
	Finding approximate solutions while computing general edit distance has also been studied extensively. The exact algorithm given in~\cite{LMS98} immediately gives a linear-time $\sqrt{n}$-approximation algorithm. A series of subsequent works improved this approximation factor first to $n^{3/7}$~\cite{BJKK04}, then to $n^{1/3+o(1)}$~\cite{BES06} and later to $2^{\widetilde{O}(\sqrt{\log n})}$~\cite{AO09} while keeping he running time of the algorithm almost linear. Batu et al.~\cite{BEKMRRS03} gave an $O(n^{1-\alpha})$-approximation algorithm that runs in time $O(n^{\alpha/2})$. The approximation factor was further improved to $(\log n)^{O(1/\epsilon)}$, for every $\epsilon>0$ by providing a $n^{1+\epsilon}$ time algorithm~\cite{AKO10}.
	

	\subsection{Our technique}
	
	To exhibit  the techniques behind our results, let us first briefly introduce 
	the main idea behind~\cite{LMS98} and~\cite{UKK85} algorithms. Both algorithms are based on computing the edit distance matrix for strings $x$ and $y$. Ukkonen~\cite{UKK85} 
	shows that in order to compute the edit distance between pairs of strings of edit 
	distance at most $k$, only $nk$ values in the matrix need to be computed, and he identifies $O(k^2)$ important entries in the matrix. Developing on that, 
	Landau et el.~\cite{LMS98}  showed that using the suffix tree machinery each of these $O(k^2)$ entries can be found in $O(1)$ 
	operations. That machinery is used in order to evaluate queries of the form: ``find the largest 
	common substring starting at positions $i$ in $x$ and $i+d$ in $y$'', where $i\in [n], d\in \pmk$. We refer
	to such queries as $\slide(d,i)$.

	Our $O(n+k^2)$ algorithm (Section~\ref{sec:streamingAlg}) uses essentially the paradigm of~\cite{LMS98}. It uses suffix trees computed for blocks
	of size $O(k)$ of the input to compose long slides from smaller slides of size $k$. This saves space, and in the case of
	large alphabets also time.
 	In our main algorithm (Section~\ref{sec:noSuffixTrees}) that runs in time $n+O(k^3)$  we do not compute suffix trees at all. Instead, 
	we implement the slide queries in the most na\"{\i}ve
	way using character by character comparison. This in general would lead to running time $O(nk)$
	as in Ukkonen's algorithm. To obtain the $n+O(k^3)$-time bound we refrain from performing 
	{\em long simultaneous} slides. By long simultaneous slides we mean slides $\slide(d,i)$ and $\slide(d',i')$
	for which the common substrings have a large overlap. Motivated by our work in~\cite{CGK16} we show
	that such slides imply periodicity of the underlying substrings. We leverage this periodicity to
	perform only one of the two slides. This generalizes to multiple simultaneous slides while
	having to pay only for one of the slides. The resulting algorithm turns surprisingly simple.
	
%
%
%
	\section{Preliminaries}
	\label{sec:prelim}

	In this section we some of the main tools we use later.
	For the sake of presentation, through out this paper we consider both the input strings to be of the same length. However one can easily generalize all the algorithms stated in this paper for two strings of different lengths.
	\paragraph{Notations:} 
	For an interval $[i,j]$ and a string $x\in \zo^n$ we denote by $x_{i,\dots,j}$ the substring 
	$x_i,\dots,x_j$ and for convenience if $i\le 0$ then $x_{i,\dots,j}=x_{1,\dots, j}$, if $j \ge n$ then $x_{i,\dots,j}=x_{i,\dots,n}$, and if $i>j$ then $x_{i,\dots,j}$ is the empty string. 
	We say that a string $x\in \Sigma^n$ is \emph{periodic} with period size $p$ if there exits a pattern $w\in \Sigma^p$, $p\le n/2$, and an integer $\ell \ge 2$ such that $x= w^{\ell} z$, where $z$ is a prefix of $w$. In such a case, we refer $w$ as \emph{period} of $x$.
		
	\subsection{Dynamic programming algorithm for computing edit distance}
	A well known dynamic programming algorithm by~\cite{WF74} solves the problem in time $O(n^2)$. 
	The algorithm proceeds by computing an $(n+1)\times (n+1)$-sized edit distance matrix $D$ indexed by $\{0,\dots,n\}\times \{0,\dots,n\}$, where the 
	$D_{i,j}$-entry stores the value $\Delta(x_{1,\ldots,i}, y_{1,\ldots,j})$ if $i>0$ and $j>0$, 
	and $\max \{i,j\}$ otherwise. The algorithm fills in
	the matrix values in a lexicographic order according to the following recurrence formula:
	\[ 
	D_{i,j} =\min \left\{\begin{array}{lr}
	D_{i-1,j}+1, & \text{if } i>0\\
	D_{i,j-1}+1, & \text{if } j>0\\
	D_{i-1,j-1}+\delta_{x_i,y_j}, & \text{if }  i,j>0.
	\end{array}\right\} 
	\]
	Where $\delta_{x_i,y_j}=0$ if $x_i=y_j$ and $1$ otherwise. The recurrence formula stems from the fact
	that an optimal alignment for $x_{1,\ldots, i}$ and $y_{1,\dots , j}$ can either (i) optimally align 
	$x_{1,\ldots, i-1}$ and $y_{1,\dots , j}$ and delete $x_{i}$, or (ii) optimally align 
	$x_{1,\ldots, i}$ and $y_{1,\dots , j-1}$ and delete $y_j$ or (iii) optimally align 
	$x_{1,\ldots, i-1}$ and $y_{1,\dots , j-1}$ and pay additional cost of $\delta(x_i,y_j)$.

	When computing the matrix in a lexicographic order 
	the values $D_{i-1,j}, D_{i,j-1} \text { and }D_{i-1,j-1}$ are already known while computing $D_{i,j}$.
	Hence each entry is evaluated in $O(1)$ operations, which implies the $O(n^2)$ bound on the total running 
	time of the algorithm.
	
	The above description can be viewed pictorially as a graph, known as \emph{edit graph}. For any two strings $x,y \in \Sigma^n$ we define the edit graph $G$ as follows: the set of vertices contains all pairs $(i,j)$ where $0 \le i,j \le n$ and the set of edges contains an edge of cost $1$ from a point $(i,j)$ to $(i+1,j)$ where $0 \le i < n$ and $0 \le j \le n $ and from a point $(i,j)$ to $(i,j+1)$ where $0 \le i \le n$ and $0 \le j < n $. The edge set also contains an edge of cost $\delta(x_i,y_j)$ from a point $(i-1,j-1)$ to $(i,j)$ where $0 < i \le n$ and $0 < j \le n $ (See Figure~\ref{fig:editgraph}). Note that in the above description, the cell $D_{i,j}$ corresponds to the point $(i,j)$ in the edit graph. The problem of computing edit distance for strings $x,y$ translates into finding the cost of a shortest path starting at $(0,0)$ and ending at $(n,n)$.
	\begin{center}
		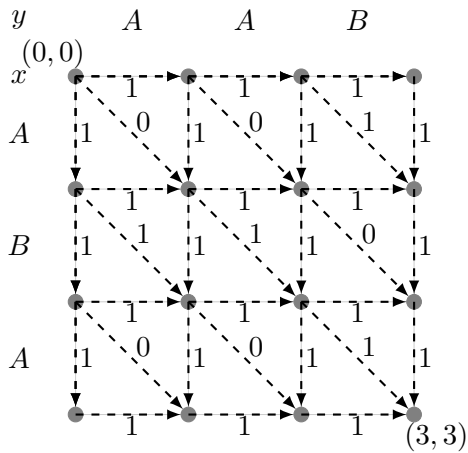
\begin{figure}[ht]
\centering
\begin{tikzpicture}[scale=1.5,shorten >=1mm,>=latex]
 \tikzstyle gridlines=[color=black!20,very thin]
 \draw[color=black!20,very thin] (0,0) grid (3,3);

 \foreach \x in {0,...,3}
  \foreach \y in {0,...,3}
   {
     \draw[fill,color=black!50] (\x,\y) circle (0.65mm);
   }
 \foreach \x in {0,...,2}
  \foreach \y in {1,...,3}
   {
     \draw[->,dashed,thick] (\x,\y)--(\x+1,\y);
     \draw[->,dashed,thick] (\x,\y)--(\x,\y-1);
     \draw[->,dashed,thick] (\x,\y)--(\x+1,\y-1);
   }
   \foreach \y in {1,...,3}
   {
     \draw[->,dashed,thick] (0,\y)--(0,\y-1);
     \draw[->,dashed,thick] (3,\y)--(3,\y-1);
     \node at (0.5,\y-0.1) {$1$};
     \node at (1.5,\y-0.1) {$1$};
     \node at (2.5,\y-0.1) {$1$};
   }
   \foreach \x in {0,...,2}
   {
     \draw[->,dashed,thick] (\x,0)--(\x+1,0);
     \draw[->,dashed,thick] (\x,3)--(\x+1,3);
     \node at (\x+0.1,2.5) {$1$};
     \node at (\x+0.1,1.5) {$1$};
     \node at (\x+0.1,0.5) {$1$};
   }
   \node at (-0.5,3) {$x$};
   \node at (-0.5,0.5) {$A$};
   \node at (-0.5,1.5) {$B$};
   \node at (-0.5,2.5) {$A$};
   \node at (-0.5,3.5) {$y$};
   \node at (0.5,3.5) {$A$};
   \node at (1.5,3.5) {$A$};
   \node at (2.5,3.5) {$B$};
   \node at (0.6,2.6) {$0$};
   \node at (1.6,2.6) {$0$};
   \node at (2.6,2.6) {$1$};
   \node at (0.6,1.6) {$1$};
   \node at (1.6,1.6) {$1$};
   \node at (2.6,1.6) {$0$};
   \node at (0.6,0.6) {$0$};
   \node at (1.6,0.6) {$0$};
   \node at (2.6,0.6) {$1$};
   \node at (0.5,-0.1) {$1$};
   \node at (1.5,-0.1) {$1$};
   \node at (2.5,-0.1) {$1$};
   \node at (3.1,2.5) {$1$};
   \node at (3.1,1.5) {$1$};
   \node at (3.1,0.5) {$1$};
   \node at (-0.2,3.2) {$(0,0)$};
   \node at (3.2,-0.2) {$(3,3)$};

\end{tikzpicture}
\caption{An example of edit graph for string $x,y \in \{A,B\}^3$}
   \label{fig:editgraph}
   
\end{figure}
	\end{center}

	\subsection{An $O(n +k^2)$-time algorithm for computing edit distance}
	\label{sec:LMS98}
	Suppose now that we are guaranteed that $\Delta(x,y) \le k$, can we derive an algorithm 
	with a better running time? Ukkonen~\cite{UKK85} provided an $O(kn)$ time algorithm by 
	realizing that in fact the algorithm does not have to compute all the matrix $D$ values but 
	only the ones that reside within the diagonals $\pmk$, 
	where the \emph{$d$-diagonal} is the set  of pairs $(i,j)$ such that $j=i+d$. 
	Sometimes we refer to the $0$-diagonal as the \emph{main diagonal}. 
	Notice, the values along each diagonal form a non-decreasing sequence of integers from range $\{0,\dots,n\}$.
	For $h\in \{0,\dots,n\}, d\in \{-h,\dots,h\}$ we define $F^h(d)$ as the furthest point on the diagonal 
	$d$ that can be reached within $h$ edit operations, formally:
	\[ F^h(d)= \max\set{i:\, D_{i,i+d} = h}.\]
	The values $F^h(d)$ fully determine the matrix, and provided that $\Delta(x,y) \le k$, the values $F^h(d)$ for $h\le k$
	determine the relevant content of the diagonals $\pmk$.
	In~\cite{LMS98} authors were able to compute each of these $F^h(d)$ values within $O(1)$ operations by 
	preprocessing the input (by building a corresponding generalized suffix tree) resulting in an 
	$O(n+k^2)$ algorithm for computing the edit distance. In the sequel, we briefly sketch their algorithm. 
	
	For a diagonal $d\in \{-h,\dots,h\}$, and for a row $i\in [n]\cup \set{0}$ we denote by $\slide_{x,y}(d,i)$ the furthest
	row $q \ge i$ that can be reached from row $i$ on diagonal $d$ while incurring no
	edit cost, formally:
	
	\[ \slide_{x,y}(d,i)= \max\{q :\, i\le q \le \min\{n,n-d\} \;\&\; x_{i+1, \ldots, q}= y_{i+d+1,\dots, q+d}\}.\]
	Notice that the slide {\em does not} compare $x_i$ against $y_i$.
        To illustrate the definitions, consider the value $F^0(0)$, this corresponds to the
	size of the largest shared prefix of $x$ and $y$, which equals $\slide_{x,y}(0,0)$.
	Furthermore, observe that whenever $x_{i+1}\neq y_{i+1}$, the value $\slide_{x,y}(d,i)$ 
	is $i$ (as $x_{i+1,...,q}=x_{i+1,i}$ which is the empty string).
	Generally, the following recurrence formula holds (see~\cite{LMS98} Lemma~2.8) for $h>0$ and $d\in\{-h,\dots,h\}$:
	\begin{eqnarray}\label{eqn:LMSrec}
	F^h(d)= \slide_{x,y} \left(d,\max S_{d,h}\right)
	\end{eqnarray}  
	\begin{eqnarray*}
\text{ where }  S_{d,h} \text{ is the set containing }\left\{\begin{array}{lr}
	F^{h-1}(d)+1, & \text{if } |d| \le h-1\\
	F^{h-1}(d+1)+1, & \text{if } d+1 \le h-1\\
	F^{h-1}(d-1), & \text{if } d-1 \ge -(h-1),
	\end{array}\right\}
	\end{eqnarray*}  
	For convenience $S^{0}(0)=\{0\}$.
	The intuition behind the recurrence is that an alignment of $x_{1,\dots, F^h(d)}$ and 
	$y_{1,\dots, F^h(d)+d}$ of cost $h$ can be obtained by one of 
	the three possibilities: 
	\begin{itemize}
		\item Optimally align $x_{1,\ldots,F^{h-1}(d)}$ and $y_{1,\ldots,F^{h-1}(d)+d}$,
		align $x_{F^{h-1}(d)+1}$ and $y_{F^{h-1}(d)+d+1}$ (this would cost an additional edit cost as these 
		values must be different) and slide on diagonal $d$, starting at $F^{h-1}(d)+1$, or
		\item optimally align $x_{1,\ldots,F^{h-1}(d+1)}$  and $y_{1,\ldots,F^{h-1}(d+1)+d+1}$,
		insert the $(1+F^{h-1}(d+1))$-th character of  $x$ and slide on diagonal $d$ starting at $F^{h-1}(d+1)+1$, or
		\item optimally align $x_{1,\ldots,F^{h-1}(d-1)}$  and $y_{1,\ldots,F^{h-1}(d-1)+d-1}$, 
		insert the $(F^{h-1}(d-1)+d)$-th character of $y$ and slide on diagonal $d$ starting at $F^{h-1}(d-1)$.
	\end{itemize} 

        We define values $c_{d,h}$ to be the maximum size of $S_{d,h}$, i.e., the number of values contributing to $S_{d,h}$ counting
multiplicity. Clearly, $c_{d,h}=3$ when $\abs{d} \le h-2$, $c_{d,h}=2$ when $\abs{d}= h-1 \ge 1$, and $c_{d,h}=1$ when $\abs{d}=h$ or $h=1$. (We define $c_{0,0}=1$
for convenience.)

	We define the $h$-wave as the set of points: $F^{h}(-h), \ldots F^{h}(0), \ldots, F^{h}(h)$.
	The algorithm proposed by~\cite{UKK85} proceeds by computing first the $0$-wave, then the 
	$1$-wave and so on. The algorithm terminates whenever it encounters a wave $e$ such that 
	$F^{e}(0)=n$. The final output of the algorithm is $\Delta(x,y)=e$. 
	
	To obtain the upper bound on the running time of the algorithm the authors in~\cite{LMS98} show that 
	the computation of $\slide_{x,y}(d,i)$ can be done in $O(1)$ operations. This is done by first preprocessing 
	the input and building a generalized suffix tree for the string $x\$y\#$ where $x,y$ are the input strings, and $\$, \#$ are characters that do not belong to the alphabet $\Sigma$ and a data structure that answer a query for lowest common ancestor for this suffix tree in $O(1)$ time. Using that data structure they are able to evaluate a query $\slide_{x,y}(d,i)$ in $O(1)$-operations, see Section 2.3 in~\cite{LMS98}.

	The above implementation of~\cite{LMS98} is done in a non-streaming fashion, since the 
	suffix tree computation requires $O(n)$ space. A natural approach to bypass this obstacle 
	is by dividing the input strings into blocks, compute a suffix tree for each block 
	separately so that we can compute the slide function on each block efficiently. However, 
	the aforementioned implementation of~\cite{LMS98} computes the $F^{h}(d)$ values in waves. 
	Therefore, if for some value of $h$
	the values $\set{F^{h}(d)}_{d\in \pmk}$ are far apart we might need to go back and forth between  
	different blocks. Thus, we first present an algorithm for computing the values $F^d(h)$ in a different
        order. In our implementation slides with smaller starting row will be computed earlier. 
	This algorithm is given in Section~\ref{sec:topDownLMS}. From this algorithm we derive a 
	streaming algorithm in Section~\ref{sec:streamingAlg}. Finally, we present our main algorithm  that does 
	not computes suffix trees at all in Section~\ref{sec:noSuffixTrees}. 
	
	We present our algorithms as computing only the edit distance. All our algorithms compute the values of $F^h(d)$ for all $\abs{d} \le h \le k$. 
From these values, one can easily reconstruct an optimal alignment of $x$ and $y$ in time $O(k)$. Storing these values requires $O(k^2)$ space. 
If we are interested only in the edit distance (the number) then Algorithms~\ref{alg:streamingLMS} and~\ref{alg:noSuffixTree} 
need only space $O(k)$ otherwise they need space $O(k^2)$.

	\section{Towards a Streaming Algorithm: a Row Modification of~\cite{LMS98}}
	\label{sec:topDownLMS}
Our goal is to design a modification of the~\cite{LMS98} algorithm that will perform all slide operations
in the order of increasing starting row. Our algorithm will determine $F^h(d)$ for all values of $d$ and $h$ such that
$\abs{d} \le h \le k$, where $k$ is a provided parameter.
To determine $F^h(d)$ using (\ref{eqn:LMSrec}), we need to take the maximum row of up-to three possible
candidate rows obtained from values of $F^{h-1}(d-1), F^{h-1}(d)$ and $F^{h-1}(d+1)$, and perform a slide on diagonal $d$ from that row.
Our algorithm will maintain $n+1$ lists, $L_0,\dots, L_n$, each list containing entries of the form $(d,h)$.
The meaning of an entry $(d,h)$ on a list $L_i$ is that the slide to compute $F^d(h)$ should possibly start at row $i$, i.e.,
$i$ is one of the three values $F^{h-1}(d-1), F^{h-1}(d)+1$ or $F^{h-1}(d+1)+1$.
At a given time, each $(d,h)$ is contained in lists $L_0,\dots, L_n$ at most once, in particular, it appears
on the list $L_i$ that corresponds to the maximum starting row $i$ for the slide of $F^h(d)$ computed thus far. An array $D$ of 
size $(2k+1) \times (k+1)$ is used to point to this unique occurrence of entry $(d,h)$ on these lists. An array $C$ of the same
dimension is used to count the number of candidate rows for the slide of $F^h(d)$ computed thus far, i.e., the number of times
entry $(d,h)$ was put onto these lists.

Additionally, the algorithm stores a $(2k+1) \times (k+1)$ array $L^{h}(d)$ of values $F^h(d)$, and a generalized suffix tree 
for the concatenation of $x$ and $y$ together with a data structure to answer the lowest common ancestor query of this suffix tree in $O(1)$ time.

%
%
%
%
	
%

	The algorithm proceeds as follows. It starts by adding the value $(0,0)$ 
	to $L_0$, and performs a slide on diagonal $0$, starting at row $0$. Now suppose that this 
	slide ended at row $q$, then the algorithm first updates $L^{0}(0)=q$ and then adds the entries $(1,1)$ into $L_{q}$ and $(0,1)$ $(-1,1)$ into $L_{q+1}$. Now for every row $i=0,\dots, n$: The algorithm scans the list $L_i$, for each entry $(d,h)$ in the list 
	it checks whether all the required values for slide $F^h(d)$ have been computed yet. If they already 
	have been computed, then it performs a slide on diagonal $d$ starting at row $i$. Assuming the 
	slide ends at row $q$, the algorithm sets $L^{h}(d)=q$ and inserts $(d+1,h+1)$ into $L_q$ and $(d,h+1)$, $(d-1,h+1)$ into 
	$L_{q+1}$. 
	
	Pseudo-code for the algorithm is below:


	\begin{algorithm}[H]
		\Input{$x,y\in \zo^n$, and a parameter $k\in [n]$ such that $\Delta_e(x,y)\le k$.}
		\Output{$\Delta_e(x,y)$ }

		\tcp{Initialization:}

		Build a generalized suffix tree for the concatenation of $x$ and $y$ in order to evaluate queries
		$slide_d(i)$ using $O(1)$ operations, as in~\cite{LMS98}. 	
		
		For $i=0,\dots,n$, initialize each list $L_i$ to be empty; 
		
		For all integers $d,h$ such that $\abs{d}  \le h \le k$, set $D(d,h)=\Null$ and $C(d,h)=0$;
		
		Invoke $\Update(0,0,0)$; 

		\tcp{Main Loop:}		

		\For{$i=0,\dots,n$}{
			
			\While{$L_i$ is not empty}{ 
				Pick $(d,h)$ from $L_i$ and remove it from the list;
				
				\If{$C(d,h)=c_{d,h}$ }  
				{
					$q=\slide_{x,y}(d,i)$;
					
					$L^h(d)=q$;
		
					\If{$h<k$}{
						$\Update(d,q+1,h+1)$;
			
						\lIf{$d<k$}
						{$\Update(d+1,q,h+1)$}
			
						\lIf{$d>-k$}
						{$\Update(d-1,q+1,h+1)$}
					}
					
				}
			}
		}

		Output the smallest $h\le k$ such that $L^h(0)=n$.
		\caption{A row modification of the~\cite{LMS98} algorithm}\label{alg:topDownLMS}
	\end{algorithm}
	
	\begin{procedure}[H]
	\hrule height .7pt\vspace{1mm}
	\TitleOfAlgo{$\Update(d,i,h)$}\label{alg:update}
	\hrule\vspace{1mm}

                Increment $C(d,h)$ by one.

		\If{$D(d,h)=\Null$ or $D(d,h)$ points to an entry in $L_{i'}$ where $i'<i$}
		{
			\If{$D(d,h)\neq \Null$}
			{Delete the current node pointed to by $D(d,h)$ from $L_{i'}$;}

			Add $(d,h)$ into the back of $L_i$;
			
			Set $D(d,h)$ to point to the new entry $(d,h)$;
		}
	\hrule\vspace{1mm}
	\end{procedure}
	
%
%
%
%

	The algorithm satisfies two key properties captured in the next lemma. 

	\begin{lemma}\label{lemma:globlaLMSCorrectness}
		Let $i\in \set{0,\dots,n}$ and let $d,h$  be integers such that $\abs{d}\le h \le k$.
\begin{enumerate}
\item $i\in S_{d,h}$ iff $(d,h)$ appears on the list $L_i$ during the run of the algorithm iff $\Update(d,i,h)$ is invoked during the run of the algorithm.
\item If $i=\max S_{d,h}$ then while processing list $L_i$, the value of $L^h(d)$ is set to $F^h(d)$. 
\end{enumerate}
	\end{lemma}

\begin{proof}
        We provide a brief sketch of an argument that proceeds by induction on $h$.
        Notice that $\Update(d,i,h)$ is only invoked for values satisfying $\abs{d}\le h \le k$. Also, the only way for $(d,h)$ to get on some list $L_i$
        is by invocation of $\Update(d,i,h)$. This proves the second 'iff' of the first part. For values $i=d=h=0$, the first property is true because
        $\Update(0,0,0)$ is invoked during the initialization, and the second property is true because $(0,0)$ is on list $L_0$ after that, $C(0,0)=c_{0,0}=1$, 
        so $\slide_{x,y}(0,0)$ will be eventually computed and $L^0(0)$ will receive the value of that slide which corresponds to $F^{0}(0)$.

        The second property claims that the value of $L^h(d)$ is set to $F^h(d)$. In order 
        to compute $F^h(d)$ we need to know up-to three values $F^{h-1}(d-1), F^{h-1}(d)+1$ and $F^{h-1}(d+1)+1$, and perform
        a slide along diagonal $d$ from their maximum. Assuming inductively that for $L^{h-1}(d-1), L^{h-1}(d)$ and $L^{h-1}(d+1)$ the second property holds,
        when the value of $L^{h-1}(d-1)$ is set to $F^{h-1}(d-1)$, $\Update(d,F^{h-1}(d-1),h)$ is invoked and indeed, $F^{h-1}(d-1)\in S_{d,h}$.
        Similarly, when $L^{h-1}(d)$ is set to $F^{h-1}(d)$, $\Update(d,F^{h-1}(d)+1,h)$ is invoked and $F^{h-1}(d)+1 \in S_{d,h}$.
        Last, when $L^{h-1}(d+1)$ is set to $F^{h-1}(d+1)$, $\Update(d,F^{h-1}(d+1)+1,h)$ is invoked and again, $F^{h-1}(d+1)+1\in S_{d,h}$.
        (These updates happen provided $\abs{d}\le h \le k$.) This means that the only way for $(d,h)$ to get on some list $L_i$ is if $i \in S_{d,h}$.
        Moreover, after the last value of the three $L^{h-1}(d-1), L^{h-1}(d)$ and $L^{h-1}(d+1)$ is computed during processing of some list $L_i$, $C(d,h)=c_{d,h}$. 
        Furthermore,  $i \le \max S_{d,h}$ as the last $\Update(d,j,h)$ happens with $i \le j \le \max S_{d,h}$. Hence, $(d,h)$ will appear on the list $L_i$ that is currently processed or 
        on some list that will be processed later. Once we reach $(d,h)$ on list $L_{\max S_{d,h}}$,  $C(d,h)=c_{d,h}$, so the $\slide_{x,y}(d,\max S_{d,h})$ is computed and $L^h(d)$
        is set to $F^h(d)$. This finishes the argument.
\end{proof}

	The correctness of the output of the algorithm follows from the second part of the lemma.

	\paragraph{Complexity analysis:}
	Let us now discuss the time complexity of Algorithm~\ref{alg:topDownLMS}. Time complexity of Algorithm~\ref{alg:topDownLMS} is determined by two main tasks. First, constructing the generalized suffix tree and lowest common ancestor data structure as in~\cite{LMS98} requires $O(n\cdot \min\{\log n, \log |\Sigma|\})$ time~\cite{Gus97}. Second, we need to process points present in the lists $L_i$ for $0\le i \le n$. As explained above, the processing of each point of the form $(d,h)$ takes $O(1)$ time. By Lemma~\ref{lemma:globlaLMSCorrectness} the total number of points added in the lists throughout the run of the algorithm is bounded by $O(k^2)$ (each point can be inserted at most $3$ times). Hence the overall time complexity of Algorithm~\ref{alg:topDownLMS} is $O(n\cdot \min\{\log n, \log |\Sigma|\}+k^2)$.
	
	Algorithm~\ref{alg:topDownLMS} requires space for three purposes. First, to construct and store the generalized suffix tree we need $O(n)$ space as in~\cite{LMS98}. Second, to maintain the lists $L_0,\dots,L_n$ and the array $D$ we use space of size $O(n+k^2)$. Third, we need to maintain all the values of $L^h(d), C(h,d)$ for $h\in \set {0, \dots, k}, d\in \pmk$ and that requires total $O(k^2)$ space. Hence total space requirement is $O(n+k^2)$.
		
	\section{An $O(n+k^2)$-time streaming algorithm for computing edit distance}
	\label{sec:streamingAlg}
	\label{slg:streamingAlg}
	In this section we prove Theorem~\ref{thm:computeEdit}. The algorithm we present is based on the algorithm described in Section~\ref{sec:topDownLMS}. The bottleneck of that algorithm was
the space needed to store the suffix tree data structure for efficient implementation of slides. To eliminate the bottleneck, we divide the input strings $x$ and $y$
	into (overlapping) blocks of length $O(k)$ and process the input block by block.
	
	For each block we will build a suffix tree data structure so that each slide operation within the block will be evaluated in $O(1)$-time. Slide operations that span
several blocks will be split into pieces of size at most $k$. When processing the $j$-th block we will process all continuing slide operations and all slide operations that start
on rows between $jk$ and $(j+1)k-1$ of the original edit distance matrix for $x$ and $y$. Instead of maintaining lists $L_0,\dots,L_n$ we will only maintain lists $L_0,\dots,L_k$ that will
contain the starting positions of slides
within the current block. The list $L_0$ will hold the slides continuing from the previous block, $L_k$ will maintain the slides that should continue into the next block. Whenever
a slide operation continues past the $(j+1)k$ row, we will put it on the list $L_k$. After finishing the current block we will move list $L_k$ to $L_0$ and we will process the next block.

	\begin{center}
		\begin{figure}[ht]
\centering
\begin{tikzpicture}[scale=.7,shorten >=.65mm,>=latex]
 \tikzstyle gridlines=[color=black!20,very thin]
 \draw[color=black!20,very thin] (0,0) grid (8,8);

 \foreach \x in {0,...,8}
  \foreach \y in {0,...,8}
   {
     \draw[fill,color=black!50] (\x,\y) circle (0.65mm);
   }
   \foreach \x in {0,...,7}
  \foreach \y in {1,...,8}
   {
     \draw[->,dashed,thick] (\x,\y)--(\x+1,\y);
     \draw[->,dashed,thick] (\x,\y)--(\x,\y-1);
     \draw[->,dashed,thick] (\x,\y)--(\x+1,\y-1);
   }
   \foreach \y in {1,...,8}
   {
     \draw[->,dashed,thick] (0,\y)--(0,\y-1);
     \draw[->,dashed,thick] (8,\y)--(8,\y-1);
   }
   \foreach \x in {0,...,7}
   {
     \draw[->,dashed,thick] (\x,0)--(\x+1,0);
     \draw[->,dashed,thick] (\x,8)--(\x+1,8);
   }



  \draw[fill=white!70!gray] (0,8) rectangle (5,5);
  \draw[fill=white!70!gray] (0,6) rectangle (7,3);
  \draw[fill=white!70!gray] (2,4) rectangle (8,1);
  \draw[fill=white!70!gray] (4,2) rectangle (8,0);
  \draw[-,dashed,thick] (0,5)--(5,5);
  \draw[-,dashed,thick] (5,6)--(5,5);
  \draw[-,dashed,thick] (2,3)--(7,3);
  \draw[-,dashed,thick] (7,4)--(7,3);
  \draw[-,dashed,thick] (4,1)--(8,1);
  \draw[-,>=latex,thick] (0,6)--(6,0);
  \draw[-,>=latex,thick] (2,8)--(8,2);
  \draw[-,>=latex,thick] (0,8)--(8,0);
  \node at (2,4.5) {$-k$};
  \node at (4,4.5) {$0$};
  \node at (6,4.5) {$k$};


\end{tikzpicture}
\caption{Sketch of a block-wise division of edit graph.}
   \label{fig:block}
\end{figure}
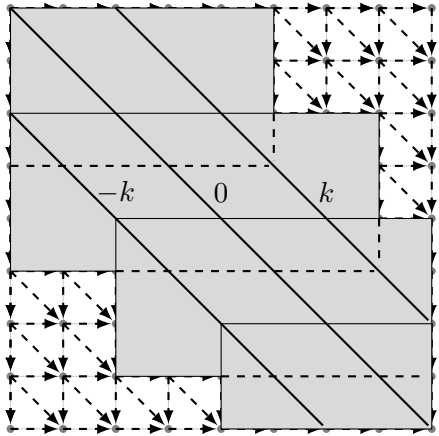

	\end{center}

       The $j$-th block of $x$ will consist of $x_{jk+1,\dots , (j+3)k+1}$ and of $y$ will be
			$y_{(j-1)k+1,\dots , (j+2)k+1}$. This provides enough context so that slides on diagonals $\{-k,\dots,k\}$ on rows between $jk$ and $(j+1)k-1$ of the original
matrix for $x$ and $y$ can be computed from slides on these blocks of $x$ and $y$ (see Fig.~\ref{fig:block}). Diagonal $d+k$ of the edit distance matrix of $x'$ and $y'$ corresponds
to diagonal $d$ of the edit distance matrix of $x$ and $y$.
      
	The pseudo-code of our algorithm is below. The procedure	$\Update()$ remains the same as in the previous section.

	\begin{algorithm}[H]
		\Input{$x,y\in \zo^n$, and a parameter $k\in [n]$ such that $\Delta_e(x,y)\le k$.}
		\Output{$\Delta_e(x,y)$ }
		
		\tcp{Initialization:}
		For $i=0,\dots,k$, initialize each list $L_i$ to be empty; 
		
		For all integers $d,h$ such that $\abs{d}  \le h \le k$, set $D(d,h)=\Null$ and $C(d,h)=0$;
		
		Invoke $\Update(0,0,0)$;

		\tcp{Main loop over blocks of size $O(k)$:}		
		\For{$j=0,\dots,\lceil n/k \rceil -1$ }
		{

			Let $x'= x_{jk+1,\dots , (j+3)k+1}$ and 
			$y'=y_{(j-1)k+1,\dots , (j+2)k+1}$;

 			\lIf{$j=0$}{ $y'=0^k \cdot y'$}

			Build a generalized suffix tree for the concatenation of $x',y'$ in order to evaluate queries
			$slide_{x',y'}(d,i)$ using $O(1)$ operations, as in~\cite{LMS98}. 	

			\For{$i=0,\dots, k-1$} 
			{  \tcp{Process each list $L_i$ within the current block}
							
				\While{$L_i$ is not empty}{ 
					Pick the first entry $(d,h)$ stored in $L_i$ and remove it from the list;

	
					\If{$C(d,h)=c_{d,h}$ }  
					{
						$q=\slide_{x',y'}(k+d,i)$;
						
						\If{$q \ge k$}{\tcp{Partial slide $\rightarrow$ the slide will continue during the next block}
							Insert $(d,h)$ into the list $L_{k}$;
							
							Set $D(d,h)$ pointing to the new entry $(d,h)$;
						}
						\Else{

							$L^h(d)=q+jk$;
			
							\If{$h<k$}{
								$\Update(d,q+1,h+1)$;
					
								\lIf{$d<k$}
								{$\Update(d+1,q,h+1)$}
				
								\lIf{$d>-k$}
								{$\Update(d-1,q+1,h+1)$}
							}
						}
					}
					
				}
			}
			Move the list $L_{k}$ to be $L_0$; \tcp{All lists except for $L_{0}$ are empty.}
		}

		Output the smallest $h\le k$ such that $L^h(0)=n$.
		\caption{A streaming algorithm for computing edit distance}
		\label{alg:streamingLMS}
	\end{algorithm}
	
	The correctness of our new algorithm follows from the correctness of the algorithm in Section~\ref{sec:topDownLMS}. The difference between the two algorithms
	lies only in dividing longer slides into smaller pieces.
		
	\paragraph{Complexity Analysis:}
	Time complexity of Algorithm~\ref{alg:streamingLMS} follows from the following claim.
	\begin{lemma}
		\label{claim:streamTime}
		In every iteration $j\in \{0,\dots, \lceil n/k \rceil-1\}$, the total number of steps performed is bounded by $O(k\cdot \min\{\log k, \log |\Sigma|\}+k_j)$, where $\sum_j k_j = O(k^2)$.
	\end{lemma}
	Now clearly the overall time complexity of Algorithm~\ref{alg:streamingLMS} is $O(n+k^2)$. It remains to prove the above claim.
	\begin{proof}[Proof of Lemma~\ref{claim:streamTime}]
		At each iteration, first we need to construct a generalized suffix tree for blocks of size $O(k)$ and a data structure for finding lowest common ancestor for that generalized suffix tree as in~\cite{LMS98} and thus we require $O(k\cdot \min\{\log k, \log |\Sigma|\})$ time~\cite{Gus97}. To set list $L_{0}$ to $L_k$, we need only $O(1)$ of pointer updates. Finally we need to process items stored in the lists $L_0,\dots,L_{k-1}$. As in the case of Algorithm~\ref{alg:topDownLMS}, processing each point takes only $O(1)$ time, and thus to conclude the proof it suffices to bound the number of items on these lists by $O(k+k_j)$. 
Let $k_j$ be the number of items that are added to lists $L_{jk+1},\dots,L_{(j+1)k-1}$ during execution of Algorithm~\ref{alg:topDownLMS} on $x$ and $y$. Clearly, those are precisely the items
that will be added to lists $L_1,\dots,L_{k-1}$ by Algorithm~\ref{alg:streamingLMS} during the $j$-th iteration. Since $L_0$ may contain at most $2k+1$ items and the same item can be added to any list
at most three times, in total we process $O(k+k_j)$ items from lists $L_0,\dots,L_{k-1}$ during the $j$-th iteration. 
The total number of points added (with multiplicity) to lists $L_0,\dots,L_n$ by Algorithm~\ref{alg:topDownLMS} is $O(k^2)$, hence $\sum_j k_j = O(k^2)$.
This completes the proof.
	\end{proof}

	Let us now discuss the space complexity of Algorithm~\ref{alg:streamingLMS}. Algorithm~\ref{alg:streamingLMS} requires space for three purposes. First, to construct generalized suffix tree for blocks of size $O(k)$ and a data structure for finding lowest common ancestor for that generalized suffix tree as in~\cite{LMS98}. This requires $O(k)$ space~\cite{Gus97}. Second, to maintain the lists $L_0,\dots,L_{k}$ and the array $D$ we use space of size $O(k^2)$. Third, we need to maintain all the values of $L^h(d)$ and $L^h(-d)$ for $h,d \in \set {0, \dots, k}$ and that requires total $O(k^2)$ space. Hence total space requirement is $O(k^2)$.
	
	\paragraph{Reducing the space requirements for computing edit distance:}
If we are interested in computing only the edit distance of $x$ and $y$ instead of their optimal alignment, we can reduce the space used by the algorithm to $O(k)$.
At any instant of time, lists $L_0,\dots,L_n$ of Algorithm~\ref{alg:topDownLMS} contain $O(k)$ items so, the same is true for lists $L_0,\dots,L_k$ of Algorithm~\ref{alg:streamingLMS}.
Indeed, at any time, for a given diagonal $d$, if $h$ is maximal such that $L^h(d)$ was already set then lists $L_0,\dots,L_n$ can contain only entries $(d,h')$ for $h'\in\{h+1,h+2\}$.
So in total the lists contain $O(k)$ items at any given time. This also means, that at any given time, arrays $C$ and $D$ have only $O(1)$ relevant entries for each $d$ so they can
be replaced by a $O(k)$-space data structure that maintains only the relevant entries and provides look-up and udate in $O(1)$ time. If we are interested only in the edit distance of 
$x$ and $y$, we do not need to store $L^h(d)$
for all possible $d$ and $h\le k$ but we can only look for the relevant entry for $d=0$. As the suffix tree data structure for efficient slides requires only $O(k)$ space in Algorithm~\ref{alg:streamingLMS}
the total space used by the algorithm can be reduced to $O(k)$.

\section{Computing Edit Distance without using Suffix Trees}
\label{sec:noSuffixTrees}

While from theoretic perspective the task of building a suffix tree requires only linear time, practically
they are quite expensive to build. Hence, for practical purposes people are using a straightforward 
implementation of ~\cite{UKK85} algorithm to compute edit distance, i.e. computing the slide function 
by comparing character by character, cf.~\cite{PP08}. In this section we propose a new approach for implementing~\cite{UKK85} algorithm, which does not involve computing suffix trees
but still avoids  long parallel slides. 

We obtain an algorithm that makes one-pass over its input $x$ and $y$, uses space $O(k)$ to compute the edit distance of $x$ and $y$
($O(k^2)$ space to compute an optimal alignment of $x$ and $y$) and runs in time $n+O(k^3)$. By writing $n+O(k^3)$ we want
to emphasize that the number of operations per an input symbol is a small constant. Indeed, to process most of the symbols
of $x$ and $y$ we perform a single comparison for each character within a simple loop. Hence, an ideal implementation
of our algorithm would just zip through most of the two strings $x$ and $y$.

Our new algorithm extends Algorithms~\ref{alg:topDownLMS} and~\ref{alg:streamingLMS} but implements all the slides in the most na\"{\i}ve way
using character by character comparison. This in general would lead to running time $O(nk)$. To bring down the cost to $n+O(k^3)$
we use ideas from~\cite{CGK16}. In particular, if we slide along two diagonals $d\le d'$ for more than $2(d'-d)$ common rows
then the corresponding two parts of $x$ and $y$ are periodic with period $d'-d$. Hence, we do not need to slide on both of the diagonals,
we may slide only on one of them and keep verifying the periodicity. Once the periodicity stops we know that at least one the two diagonals
must end its slide and pay an edit operation. For two diagonals sliding in parallel, this does not give much of savings but this naturally
generalizes to sliding along multiple diagonals $d_1 < d_2 < \cdots < d_\ell$ in parallel. If they slide in parallel for more than $2(d_\ell-d_1)$ rows then
the corresponding parts of $x$ and $y$ are periodic with period $gcd(d_\ell-d_1,d_\ell-d_2,\dots,d_\ell-d_{\ell-1})$. Again, it suffices to slide along
only one of them (most conveniently along the rightmost one) and keep verifying its periodicity (see Figure\ref{fig:matureDiag}). Once the periodicity stops,
either the rightmost diagonal has to pay an edit operation, all of them have to pay, or all but the rightmost one.
As the average length of a slide is $n/k \gg 4k$, for $k \ll \sqrt{n}$, this gives a noticeable advantage.
\begin{center}
	\begin{figure}[ht]
\centering
\begin{tikzpicture}[scale=0.3,shorten >=1mm,>=latex]
 \tikzstyle gridlines=[color=black!20,very thin]

\draw[color=black!20,thin] (-1,9) rectangle (18,21);
\draw[-,thick](0,20)--(1,19);
\draw[fill,color=gray!50] (1,19) circle (2mm);
\draw[-,thick](1,19)--(3,17);
\draw[fill,color=gray!50] (3,17) circle (2mm);

\draw[fill,color=gray!50] (1,18) circle (2mm);
\draw[-,thick](1,18)--(5,14);
\draw[-,dashed,thick](1,18)--(9,10);
\draw[-,thick](1,20)--(5,16);
\draw[-,dashed,thick](5,16)--(11,10);
\draw[-,thick](5,20)--(10,15);
\draw[-,dashed,thick](10,15)--(15,10);
\draw[-,thick](6,20)--(7,19);
\draw[fill,color=gray!50] (8,19) circle (2mm);
\draw[-,thick](8,19)--(17,10);
\draw[fill,color=gray!50] (7,19) circle (2mm);
\draw[-,thick](7,19)--(9,17);
\draw[fill,color=gray!50] (9,17) circle (2mm);

\draw[fill,color=gray!50] (9,10) circle (2mm);
\draw[fill,color=gray!50] (11,10) circle (2mm);
\draw[fill,color=gray!50] (15,10) circle (2mm);
%

\end{tikzpicture}
\caption{Illustration of the slides performed by Algorithm~\ref{alg:noSuffixTree}: Diagonals on which the algorithm compares character by character are marked by continuous lines; Mature diagonals which are not the right most are marked by dashed lines;  Edit operations are marked by a circle.}
   \label{fig:matureDiag}
\end{figure}
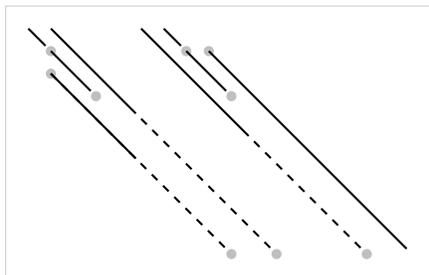
\end{center}

We extend our previous algorithms using this idea. 
We will say that a diagonal $d$ is {\em mature} at row $i$
if $x_{i-4k+1,\dots, i }= y_{i-4k+1+d,\dots, i+d}$. If we have two or more diagonals that are mature at row $i$ then we know that the 
corresponding parts of $x$ and $y$ are periodic. Our algorithm mimics Algorithm~\ref{alg:streamingLMS} in the way that it breaks
each slide operation into atomic pieces of just one character slides. The algorithm maintains lists $L_0,L_1,\dots,L_n$
that keep track of the sliding diagonals as in Algorithm~\ref{alg:streamingLMS}. (At any given time, only two lists $L_i$ and $L_{i+1}$ are
non-empty. We process diagonals on list $L_i$ in turn which puts some diagonals on the next list $L_{i+1}$.)
For each diagonal on list $L_i$
we also keep track of when it got on a list for the first time, so we extend our entries $(d,h)$ to $(d,h,\startSlide)$.
If a diagonal slides for more than $4k$ steps, it is put on a special list of mature diagonals $\matureList$ where it will
stay up until a mismatch on this diagonal occurs. 
(The mature diagonals will only require little attention most of the time
and they are handled in $\handleMatures(i)$.)
Procedure $\handleMatures()$ maintains the length of the current period of $x$ and $y$ 
in $\maturePeriod$, and $\rightMature$ stores the index of the rightmost diagonal on $\matureList$.
In addition to that, our algorithm maintains arrays $C$ and $D$ that have the same meaning as in Algorithm~\ref{alg:streamingLMS}.

The algorithm processes the input strings $x$ and $y$ row by row of the edit distance matrix of $x$ and $y$. At row $i$ it performs two main tasks: 
First, it deals with the mature diagonals encountered so far. In this part,
the algorithm checks whether there is a mismatch on the right most mature 
diagonal and whether the input strings respect the periodic pattern implied 
by mature diagonals. If both the conditions are met,
then no edit operation is taken at row $i$ for any of the mature 
diagonals. Otherwise, the algorithm identifies which are the mature diagonals 
that pay an edit operation and migrates them to the corresponding list $L_i$.

Second, the algorithm processes the list $L_i$. 
The algorithm first checks whether the current entry $(d,h,\startSlide)$ is {\em definite}: 
if it is not, then it is discarded. (The entry is {\em definite} if $\startSlide = \max S_{d,h}$ which happens iff $C(d,h)=c_{d,h}$.) Otherwise, 
the algorithm checks for a mismatch in row $i+1$ on diagonal $d$. In the case that there is 
a mismatch it mimics the behavior of Algorithm~\ref{alg:topDownLMS}. Otherwise, 
it checks whether the current slide on diagonal $d$ is long enough 
(by checking whether $i-\startSlide\ge 4k$) and if it is then the algorithm
migrates this entry into the mature diagonals list. Otherwise it is migrated to the 
list $L_{i+1}$.

Before we give a pseudo-code of the algorithm let us explain how we determine
the periodicity implied by the mature diagonals. Our key lemma, see Corollary~\ref{claim:matureImpliesPeriodic} below, 
asserts that for every row $i$,
if diagonals $d<d'$ are mature with respect to row $i$, then $x_{i-4k+1,\dots, i}$ and $y_{i-4k+1+d,\dots, i+d}$ are 
periodic with period $d'-d$. Since this is true for every pair of mature diagonals, using properties of periodicity, 
this implies that the corresponding substrings of $x$ and $y$ are periodic with period size 
$p=gcd\{d_\ell-d:\, d\in M\}$, where $M$ is the set of mature
diagonals and $d_\ell$ is the right most one. A pseudo-code of the algorithm follows.
\begin{algorithm}[H]
	\Input{$x,y\in \zo^n$, and a parameter $k\in [n]$ such that $\Delta_e(x,y)\le k$.}
	\Output{$\Delta_e(x,y)$ }
	
	\tcp{Initialization:}
	Initialize all lists $L_i$ and $\matureList$ to be empty; 
	
	For all integers $d,h$ such that $\abs{d}  \le h \le k$, set $D(d,h)=\Null$ and $C(d,h)=0$;
	
	Set $\rightMature=\infty$ and $\maturePeriod = 0$ ; 

	Invoke $\Update(0,0,0)$; 

	\tcp{Main Loop:}		
		
	\For{$i=0,\dots, n-1 $}
	{
		
		$\handleMatures(i)$;

		\While{$L_i$ is not empty}{ 
			Pick the first entry $(d,h,\startSlide)$ stored in $L_i$ and remove it from the list;
						
			\If{$C(d,h)= c_{d,h}$} 
			{
				\If{$x_{i+1} \neq y_{i+d+1}$ or $i+d+1 > n$}{ \tcp{Diagonal $d$ just finished sliding}
					$L^h(d)=i$;
		
					\If{$h<k$}{
						$\Update(d,i+1,h+1)$;
			
						\lIf{$d<k$}
						{$\Update(d+1,i,h+1)$}
			
						\lIf{$d>-k$}
						{$\Update(d-1,i+1,h+1)$}
					}
					
				}
				\Else
				{
					\lIf{$i-\startSlide > 4k$}
					{$\moveToMatures(d,h,\startSlide)$}
					\lElse{ Insert $(d,h,\startSlide)$ into $L_{i+1}$}	
				}
				
			}

		}
		
	}
		
	Output the smallest $h\le k$ such that $L^h(0)=n$.
	\caption{An algorithm for computing edit distance that does not use suffix trees}\label{alg:noSuffixTree}
\end{algorithm}

\begin{procedure}[ht]
	\hrule height .7pt\vspace{1mm}
	\TitleOfAlgo{$\newRightMature(d)$}
	\hrule\vspace{1mm}
	Set $\rightMature=d$ and $\maturePeriod = 0$;
	
	\tcp{Recompute the period}
	\ForEach{$(d',h',\startSlide') \in \matureList$}{
		
		\lIf{$d\neq d'$}{$\maturePeriod =gcd (\maturePeriod, d-d')$}
	}
		\hrule
\end{procedure}

\begin{procedure}[H]
	\hrule height .7pt\vspace{1mm}
	\TitleOfAlgo{$\moveToMatures(d,h,\startSlide)$}
	\hrule\vspace{1mm}
	Add $(d,h,\startSlide)$ into $\matureList$, and let $D(d,h)$ point to that entry;
	
	\lIf{$\matureList$ contains only one entry}{Set $\rightMature=d$ and $\maturePeriod=0$}
	\Else
	{
		\lIf{$d>\rightMature$}
		{ 
			 $\newRightMature(d)$
		}
		\lElse {
			$\maturePeriod = gcd(\maturePeriod, \rightMature-d)$
		}

	}
\end{procedure}

\begin{procedure}[H]
	\hrule height .7pt\vspace{1mm}
	\TitleOfAlgo{$\handleMatures(i)$}
	\hrule\vspace{1mm}
	\If{$\matureList$ is non-empty}{
		\If{$\matureList$ contains only one entry}{ 

			\If{$x_{i+1} \neq y_{i+\rightMature+1}$}
 			{  \tcp{The unique mature diagonal has a mismatch}
				Move the entry of diagonal $\rightMature$ from $\matureList$ to $L_{i}$; 				
			}

			\Return;
		}

		\If{$x_{i+1}= x_{i+1-\maturePeriod}$}
		{
			\If{$x_{i+1} \neq y_{i+\rightMature +1}$}
			{   \tcp{The rightmost mature diagonal has a mismatch but none else}
				
				Move the entry of diagonal $\rightMature$ from $\matureList$ to $L_{i}$; 
				
				Find $d$, the current largest diagonal in $\matureList$, and invoke
				$\newRightMature(d)$;
				
			}

			\Return;
			
		}
			
		\tcp{Mismatch on all but possibly the rightmost mature diagonals}
%
%
	
		\ForEach{$(d,h,\startSlide)\in \matureList$}
		{
			\If{$d\neq \rightMature$}{ Move the entry $(d,h,\startSlide)$ from $\matureList$ to $L_{i}$; }
		}
		
		\If{$x_{i+1}\neq y_{i+\rightMature+1}$}
		{
			\tcp{Mismatch also on the rightmost diagonal}
			Move the entry of diagonal $\rightMature$ from $\matureList$ to $L_{i}$; 
		
		}
		Set $\maturePeriod=0$;
	}
	\hrule
\end{procedure}

The procedure $\Update(d,i,h)$ is as in Algorithm~\ref{alg:topDownLMS} except that instead of inserting $(d,h)$ to the list $L_i$, 
it inserts $(d,h,i)$ into that list.
For convenient we define $gcd(0,a)=a$.

\paragraph{Correctness of Algorithm~\ref{alg:noSuffixTree}.}
To prove the correctness of our algorithm we will need the following standard facts (cf.~\cite{CR94}).

\begin{proposition}[Cf.~\cite{CR94}]\label{claim:gcd}
\begin{enumerate}
\item	Let $x\in \{0,1\}^*$ be a string and assume $x$ is periodic with period size $p$ and $q$.
	Then $x$ is periodic with period size $gcd(p,q)$.
\item  Let $w,u,v\in \{0,1\}^*$ be such that $vw=wu$ and $|v|=|u|\le |w|$ . Then $vw$ is periodic with a period of size $|v|$.
\item  Let $w,u,v\in \{0,1\}^*$ be such that $vw$ and $wu$ are periodic with a period of the same size $\le |v|,|u| \le |w|$. Then $vwu$ 
	is also periodic with a period of the same size.
\end{enumerate}
\end{proposition}

We will make use of the following simple corollary:

\begin{corollary}\label{claim:matureImpliesPeriodic}
	Let $x,y\in \Sigma^n$.
	Let $d'>d\in \pmk$ be diagonals, let $2(d'-d) \le m \le i \le n$. 
	If \[ x_{i-m+1,\dots, i }= y_{i-m+1+d,\dots, i+d } \text{ and } 
	x_{i-m+1,\dots, i }= y_{i-m+1+d',\dots, i+d' } \]
	Then both $x_{i-m+1,\dots, i }$ and $y_{i-m+1+d,\dots, i+d}$ are periodic with a period of size $d-d'$.
\end{corollary}	

\begin{proof}
Set $v=y_{i-m+1+d,\dots, i-m+1+d'-1}$, $w=y_{i-m+1+d',\dots,i+d}$ and $u=y_{i+d+1,\dots,i+d'}$. Apply the second part of the previous proposition.
\end{proof}

We will need the following main technical lemma about mature diagonals.

\begin{lemma}\label{cor:matures}
	Let $x,y\in \Sigma^n$.
	Let $i$ be an integer so that $4k \le i <n$. Let $M = \{d_1<d_2<\cdots <d_\ell\} \subseteq \pmk$ be such that for each $d\in M$, $x_{i-4k+1,\dots ,i} = y_{i-4k+1+d,\dots ,i+d}$. Let $p=gcd\{d_\ell-d:\, d\in M\}$. Then: 

	\begin{enumerate}
		\item $x_{i-4k+1,\dots ,i}$ and $y_{i-4k+1+d_1,\dots ,i+d_\ell}$ are periodic 
		with period size $p$.


		\item If $x_{i+1}=x_{i+1-p}$ then for all $j\in\{1,\dots,\ell-1\}$, $x_{i+1} = y_{i+1+d_j}$.

		\item If $x_{i+1}\neq x_{i+1-p}$ then for all $j\in\{1,\dots,\ell-1\}$, $x_{i+1} \neq y_{i+1+d_j}$.
	\end{enumerate}
\end{lemma}	

\begin{proof}
	For the first part. By applying the previous corollary for $d'=d_\ell$ and $d=d_j$, $j=1,\dots,\ell-1$, we get that $x_{i-4k+1,\dots ,i}$
 	is periodic with each period size $d_\ell-d_j$. By the first part of Proposition~\ref{claim:gcd}, $x_{i-4k+1,\dots ,i}$ is periodic
	with a period of size $p$, so each $y_{i-4k+1+d,\dots ,i+d}$, $d\in M$, is periodic with a period of size $p$. By repeated application 
	of the third part of Proposition~\ref{claim:gcd}, $y_{i-4k+1+d_1,\dots ,i+d_\ell}$ is periodic with a period of size $p$.

	For the second and third parts. By the first part, $y_{i-4k+1+d_1,\dots ,i+d_\ell}$ is periodic with a period of size $p$. Since $d_j < d_\ell$, we get
        $y_{i+1+d_j}=y_{i+1+d_j-p}$. By the assumption, $x_{i+1-p}=y_{i+1-p+d_j}$, so $y_{i+1+d_j}=x_{i+1-p}$. Hence, $x_{i+1} = x_{i+1-p}$
	iff $x_{i+1} = y_{i+1+d_j}$.
\end{proof}

When running Algorithm~\ref{alg:noSuffixTree} on strings $x$ and $y$ we can assert several properties.

 \begin{claim}\label{claim:mainLoop}
	Let $i \in [n]$, at beginning of the $i$-th iteration, if $(d,h,\startSlide)$ is on list $L_i$, $m=i-\startSlide \le 4k$
and $C(d,h)=c_{d,h}$ then $x_{i-m+1,\dots,i}=y_{i-m+1+d,\dots,i+d}$.
\end{claim}

The claim follows from an easy inspection of the main loop of the algorithm. The next combinatorial property justifies our handling 
of mature diagonals.
 
\begin{claim}\label{claim:algFlow}
	Let $i \in [n]$, during the $i$-th iteration, after invoking the procedure $\handleMatures(i)$,  
	the following holds: 
	\begin{enumerate}
		\item The list $\matureList$ consists of diagonals $d$ that were on this list at the end of iteration $i-1$ and for 
		which $x_{i+1}=y_{i+1+d}$. For such diagonals it holds that $x_{i-4k+1,\dots, i+1}= y_{i-4k+1+d,\dots, i+1+d}$. 
		Diagonals $d$ that were stored in $\matureList$ on previous iteration for which
		$x_{i+1}\neq y_{i+d+1}$ were migrated to $L_i$.
		\item $\rightMature$ stores the value of the largest diagonal stored in $\matureList$.
		\item If $\matureList$ contains at least two entries then $\maturePeriod = gcd\{\rightMature - d:\, d\in \matureList\}$.
	\end{enumerate} 
	Moreover, after invoking the procedure $\moveToMatures(d,i,h)$ items 2-3 still hold, and
	the list $\matureList$ contains only mature diagonals.
\end{claim}

\begin{proof}
For the second and third property. $\rightMature$ is updated whenever the rightmost diagonal leaves $\matureList$ or a new rightmost diagonal enters the list. Similarly, $\maturePeriod$ is updated to the claimed value when either the rightmost diagonal changes, a new diagonal enters $\matureList$, 
or all the other diagonals leave the list because of a mismatch.

For the first part. By Claim~\ref{claim:mainLoop} diagonal $d$ satisfies $x_{i-4k+1,\dots, i+1}= y_{i-4k+1+d,\dots, i+1+d}$ when it is moved
to the $\matureList$ in the main loop. Then it maintains this property inductively: If $d$ is the rightmost diagonal then it remains
on the $\matureList$ at iteration $i$ if $x_{i+1} = y_{i+1+\rightMature}$, and it is removed from the list otherwise. If $d$ is not the rightmost diagonal
then at iteration $i$ either $x_{i+1}= x_{i+1-\maturePeriod}$ or not. If $x_{i+1}= x_{i+1-\maturePeriod}$ then  the diagonal stays on $\matureList$ for the next iteration and by Lemma~\ref{cor:matures}, $x_{i+1} = y_{i+1+d}$. On the other hand, if $x_{i+1} \neq x_{i+1-\maturePeriod}$ then
the diagonal is moved from $\matureList$ to $L_{i}$ at iteration $i$, and by Lemma~\ref{cor:matures} it is the case that $x_{i+1} \neq y_{i+1+d}$.
\end{proof}

From the above we can conclude that diagonals are on $\matureList$ only when they are sliding. Once they stop sliding they are moved back
to list $L_i$. List $L_i$ maintains sliding diagonals for up-to $4k$ steps of each slide and then it moves them to $\matureList$ where they
continue sliding, or they end their slide. Algorithm~\ref{alg:noSuffixTree} mimics in this way the behavior of Algorithm~\ref{alg:topDownLMS}.

\subsection{Complexity Analysis}
\begin{lemma}
	Let $x,y\in \zo^n$, be such that $\Delta_e(x,y)\le k$, then Algorithm~\ref{alg:noSuffixTree} 
	computes $\Delta_e(x,y)$  in time $O(n + k^3)$ using $O(k)$ space, and can compute the optimal alignment using $O(k^2)$ extra space.
	The algorithm can be modified  to never run in time more than $O(kn)$.
\end{lemma}	

\begin{proof}
	First, consider the procedure $\newRightMature()$. In total this procedure can be invoked at most $O(k^2)$ times as
	every diagonal can become a mature diagonal at most $k$ times and we have $2k+1$ different diagonals. The procedure
	has to compute the greatest common divisor of up-to $2k+1$ numbers from the range $\{0,\dots,2k\}$. If we use Euclid's algorithm
	then for each diagonal the algorithm either finishes in constant time or decreases the greatest common divisor by at least one. So the
	the total number of steps spent in the $gcd$ computation is $O(k)$, and hence the procedure always finishes its computation with $O(k)$ steps. 
	Thus the total time spent in procedure $\newRightMature()$ is bounded by $O(k^3)$.

	Similar argument gives that the algorithm spends in total at most $O(k^3)$ steps in procedure $\moveToMatures()$.

	Consider the procedure $\handleMatures()$. The procedure runs in constant time whenever there is no need 
	to remove any diagonal from $\matureList$. Otherwise, it runs in time $O(\matureList) \le O(k)$ (not counting time spent in $\newRightMature()$).
	In this latter case we remove at least one diagonal from  $\matureList$ which might happen at most $O(k^2)$ times. So the total
	time spent in $\handleMatures()$ when removing some diagonal is bounded by $O(k^3)$. When not removing any diagonals we spend there $O(n)$ steps.
	
	As for the main loop. The main loop will perform $n$ iterations. Each iteration may involve slides of several diagonals by one step.
	However, each diagonal in its given slide can slide step-wise in the main loop for at most $4k$ iterations, and then it is moved to $\matureList$.
	As each of the $2k+1$ diagonals undergoes at most $k$ slides, the total number of slide steps performed in the main loop is bounded by $O(k^3)$.
	
	Thus the algorithm runs in time $O(n+k^3)$, and one can verify that the amount of work needed per one of the $n$ symbols is tiny when not
	contributing to the $O(k^3)$ time bound.

	The algorithm can be modified so that it never runs in time more than $O(kn)$. Indeed, for each $n$ we have to process at most $O(k)$ diagonals, and except 
	for computing $gcd$ each of them costs $O(1)$ operations.
	Asymptotically the most time consuming part appears to be recomputing $gcd$ after every change in the rightmost mature diagonal. But this needs to be done
	at most once for each of the $n$ rows. As computing the $gcd$ for a given row will take at most $O(k)$ time the total running time is $O(nk)$.

	The space requirements are similar to Algorithm~\ref{alg:streamingLMS}, except that we do not need to build the suffix tree data structure.
	At each iteration $i\in [n]$ the algorithm maintains only non-empty lists $\matureList, L_i$ and $L_{i+1}$, maintains arrays $C$ and $D$,
	and accesses symbols $x_{i-2k,\dots,i+1}$ and $y_{i-k,\dots, i+k+1}$ from the input strings. The arrays $C$ and $D$ can be stored in $O(k)$ space 
	as only $O(k)$ entries need to be preserved at any given time (similarly to Algorithm~\ref{alg:streamingLMS}). If we are interested only
	in computing the edit distance but not an optimal alignment of $x$ and $y$ we do not have to store all the computed values of $L^h(D)$ so we need only $O(k)$ space in total.
	Otherwise we need $O(k^2)$ space to store all the values of $L^h(d)$.
\end{proof}

\section{Disussion and further improvements}

An optimal implementation of our algorithm can represent linked lists using fixed arrays of size $O(k)$ so that there is no need to
allocate and deallocate memory for each individual item. Most of the time, more than $n - O(k^3)$ steps, the algorithm spends
sliding along the rightmost mature diagonal as there are only mature diagonals sliding at those steps. An optimal implementation
should take this into account, and it should be centered around sliding the rightmost mature diagonal. Concieveably, the sliding could
be sped up by precomputed hashes of various substrings.

Algorithms~\ref{alg:streamingLMS} and~\ref{alg:noSuffixTree} can be combined to get an algorithm running in time $\min(O(n+k^2),n+O(k^3))$:
if $O(n+k^2) \le n+O(k^3)$ run the former algorithm otherwise run the latter one. Algorithm~\ref{alg:noSuffixTree} can also be modified
to build a suffix tree data structure for slides for a block of the next $4k$ rows like Algorithm~\ref{alg:streamingLMS} whenever
a new diagonal starts sliding. When there were only mature diagonals we would run as Algorithm~\ref{alg:noSuffixTree}. This would again achieve
time complexity $\min(O(n+k^2),n+O(k^3))$, and perhaps even $n+O(k^2)$ assuming a certain combinatorial properties of edit distance matrices were true.


Another avenue to design an algorithm with time complexity $n+\tilde O(k^2)$ is by generalizing the idea of mature diagonals to
diagonals of various age groups sliding for between $2^m$ and $2^{m+1}$ steps. Diagonals within a given age group would be treated
similarly to mature diagonals in Algorithm~\ref{alg:noSuffixTree}. For each age group $2^m\le \ell\le 2^{m+1}$, we divide the diagonals 
into equal segments of size $O(\ell)$. Diagonals within a given age group belonging to the same segment would be treated
as mature diagonals in Algorithm~\ref{alg:noSuffixTree}. In such a way we refrain from parallel sliding not only for slides of 
length $\Omega(k)$, but rather for any slide length, provided that the distance between the diagonals is small. 
We believe that this achieves running time $n+\tilde O(k^2)$ albeit for the cost of more complex code. 
We plan to include this modification in the full version of this paper.

	\bibliographystyle{amsalpha}
	\bibliography{editDistance}

\newcommand{\etalchar}[1]{$^{#1}$}
\providecommand{\bysame}{\leavevmode\hbox to3em{\hrulefill}\thinspace}
\providecommand{\MR}{\relax\ifhmode\unskip\space\fi MR }
\providecommand{\MRhref}[2]{%
  \href{http://www.ams.org/mathscinet-getitem?mr=#1}{#2}
}
\providecommand{\href}[2]{#2}
\begin{thebibliography}{BYJKK04}

\bibitem[AKO10]{AKO10}
Alexandr Andoni, Robert Krauthgamer, and Krzysztof Onak, \emph{Polylogarithmic
  approximation for edit distance and the asymmetric query complexity}, 51th
  Annual {IEEE} Symposium on Foundations of Computer Science, {FOCS} 2010,
  October 23-26, 2010, Las Vegas, Nevada, {USA}, 2010, pp.~377--386.

\bibitem[AO09]{AO09}
Alexandr Andoni and Krzysztof Onak, \emph{Approximating edit distance in
  near-linear time}, Proceedings of the Forty-first Annual ACM Symposium on
  Theory of Computing (New York, NY, USA), STOC '09, ACM, 2009, pp.~199--204.

\bibitem[BEK{\etalchar{+}}03]{BEKMRRS03}
Tugkan Batu, Funda Erg\"{u}n, Joe Kilian, Avner Magen, Sofya Raskhodnikova,
  Ronitt Rubinfeld, and Rahul Sami, \emph{A sublinear algorithm for weakly
  approximating edit distance}, Proceedings of the Thirty-fifth Annual ACM
  Symposium on Theory of Computing (New York, NY, USA), STOC '03, ACM, 2003,
  pp.~316--324.

\bibitem[BES06]{BES06}
Tu\u{g}kan Batu, Funda Ergun, and Cenk Sahinalp, \emph{Oblivious string
  embeddings and edit distance approximations}, Proceedings of the Seventeenth
  Annual ACM-SIAM Symposium on Discrete Algorithm (Philadelphia, PA, USA), SODA
  '06, Society for Industrial and Applied Mathematics, 2006, pp.~792--801.

\bibitem[BI15]{BI15}
Arturs Backurs and Piotr Indyk, \emph{Edit distance cannot be computed in
  strongly subquadratic time (unless {SETH} is false)}, Proceedings of the
  Forty-Seventh Annual ACM on Symposium on Theory of Computing (New York, NY,
  USA), STOC '15, ACM, 2015, pp.~51--58.

\bibitem[BK15]{BK15}
Karl Bringmann and Marvin K{\"{u}}nnemann, \emph{Quadratic conditional lower
  bounds for string problems and dynamic time warping}, {IEEE} 56th Annual
  Symposium on Foundations of Computer Science, {FOCS} 2015, Berkeley, CA, USA,
  17-20 October, 2015, 2015, pp.~79--97.

\bibitem[BYJKK04]{BJKK04}
Z.~Bar-Yossef, T.S. Jayram, R.~Krauthgamer, and R.~Kumar, \emph{Approximating
  edit distance efficiently}, Foundations of Computer Science, 2004.
  Proceedings. 45th Annual IEEE Symposium on, Oct 2004, pp.~550--559.

\bibitem[BZ16]{BZ16}
Djamal Belazzougui and Qin Zhang, \emph{Edit distance: Sketching, streaming and
  document exchange}, In Proc. of {IEEE} Symposium on Foundations of Computer
  Science {(FOCS} 2016), 2016, p.~to appear.

\bibitem[CGK16]{CGK16}
Diptarka Chakraborty, Elazar Goldenberg, and Michal Kouck{\'{y}},
  \emph{Streaming algorithms for embedding and computing edit distance in the
  low distance regime}, Proceedings of the 48th Annual {ACM} {SIGACT} Symposium
  on Theory of Computing, {STOC} 2016, Cambridge, MA, USA, June 18-21, 2016,
  2016, pp.~712--725.

\bibitem[CLL{\etalchar{+}}11]{CLLPTZ11}
Ho{-}Leung Chan, Tak~Wah Lam, Lap{-}Kei Lee, Jiangwei Pan, Hing{-}Fung Ting,
  and Qin Zhang, \emph{Edit distance to monotonicity in sliding windows},
  Algorithms and Computation - 22nd International Symposium, {ISAAC} 2011,
  Yokohama, Japan, December 5-8, 2011. Proceedings, 2011, pp.~564--573.

\bibitem[CR94]{CR94}
Maxime Crochemore and Wojciech Rytter, \emph{Text algorithms}, Oxford
  University Press, 1994.

\bibitem[EJ08]{EJ08}
Funda Erg{\"{u}}n and Hossein Jowhari, \emph{On distance to monotonicity and
  longest increasing subsequence of a data stream}, Proceedings of the
  Nineteenth Annual {ACM-SIAM} Symposium on Discrete Algorithms, {SODA} 2008,
  San Francisco, California, USA, January 20-22, 2008, 2008, pp.~730--736.

\bibitem[Fre75]{Fred75}
Michael~L. Fredman, \emph{On computing the length of longest increasing
  subsequences}, Discrete Mathematics \textbf{11} (1975), no.~1, 29 -- 35.

\bibitem[GG07]{GG07}
Anna G{\'{a}}l and Parikshit Gopalan, \emph{Lower bounds on streaming
  algorithms for approximating the length of the longest increasing
  subsequence}, 48th Annual {IEEE} Symposium on Foundations of Computer Science
  {(FOCS} 2007), October 20-23, 2007, Providence, RI, USA, Proceedings, 2007,
  pp.~294--304.

\bibitem[GJKK07]{GJKK07}
Parikshit Gopalan, T.~S. Jayram, Robert Krauthgamer, and Ravi Kumar,
  \emph{Estimating the sortedness of a data stream}, Proceedings of the
  Eighteenth Annual {ACM-SIAM} Symposium on Discrete Algorithms, {SODA} 2007,
  New Orleans, Louisiana, USA, January 7-9, 2007, 2007, pp.~318--327.

\bibitem[Gus97]{Gus97}
Dan Gusfield, \emph{Algorithms on strings, trees, and sequences - computer
  science and computational biology}, Cambridge University Press, 1997.

\bibitem[Lev66]{Lev65}
VI~Levenshtein, \emph{{Binary Codes Capable of Correcting Deletions, Insertions
  and Reversals}}, Soviet Physics Doklady \textbf{10} (1966), 707.

\bibitem[LMS98]{LMS98}
Gad~M. Landau, Eugene~W. Myers, and Jeanette~P. Schmidt, \emph{Incremental
  string comparison}, SIAM J. Comput. \textbf{27} (1998), no.~2, 557--582.

\bibitem[LVZ05]{LVZ05}
David Liben{-}Nowell, Erik Vee, and An~Zhu, \emph{Finding longest increasing
  and common subsequences in streaming data}, Computing and Combinatorics, 11th
  Annual International Conference, {COCOON} 2005, Kunming, China, August 16-29,
  2005, Proceedings, 2005, pp.~263--272.

\bibitem[McC76]{Mc76}
Edward~M. McCreight, \emph{A space-economical suffix tree construction
  algorithm}, J. ACM \textbf{23} (1976), no.~2, 262--272.

\bibitem[MP80]{MP80}
William~J. Masek and Michael~S. Paterson, \emph{A faster algorithm computing
  string edit distances}, Journal of Computer and System Sciences \textbf{20}
  (1980), no.~1, 18 -- 31.

\bibitem[Nav01]{Nav01}
Gonzalo Navarro, \emph{A guided tour to approximate string matching}, ACM
  Comput. Surv. \textbf{33} (2001), no.~1, 31--88.

\bibitem[PP08]{PP08}
Dimitrios~P. Papamichail and Georgios~P. Papamichail, \emph{Improved algorithms
  for approximate string matching (extended abstract)}, CoRR
  \textbf{abs/0807.4368} (2008).

\bibitem[Sch61]{Sch61}
C.~Schensted, \emph{Longest increasing and decreasing subsequences}, Canadian
  Journal of Mathematics \textbf{13} (1961), 179--191.

\bibitem[SS13]{SS13}
Michael~E. Saks and C.~Seshadhri, \emph{Space efficient streaming algorithms
  for the distance to monotonicity and asymmetric edit distance}, Proceedings
  of the Twenty-Fourth Annual {ACM-SIAM} Symposium on Discrete Algorithms,
  {SODA} 2013, New Orleans, Louisiana, USA, January 6-8, 2013, 2013,
  pp.~1698--1709.

\bibitem[SW07]{SW07}
Xiaoming Sun and David~P. Woodruff, \emph{The communication and streaming
  complexity of computing the longest common and increasing subsequences},
  Proceedings of the Eighteenth Annual {ACM-SIAM} Symposium on Discrete
  Algorithms, {SODA} 2007, New Orleans, Louisiana, USA, January 7-9, 2007,
  2007, pp.~336--345.

\bibitem[Ukk85]{UKK85}
Esko Ukkonen, \emph{Algorithms for approximate string matching}, Inf. Control
  \textbf{64} (1985), no.~1-3, 100--118.

\bibitem[Ukk95]{Ukkonen95}
Esko Ukkonen, \emph{On-line construction of suffix trees}, Algorithmica
  \textbf{14} (1995), no.~3, 249--260.

\bibitem[Wei73]{Weiner73}
Peter Weiner, \emph{Linear pattern matching algorithms}, Proceedings of the
  14th Annual Symposium on Switching and Automata Theory (Swat 1973)
  (Washington, DC, USA), SWAT '73, IEEE Computer Society, 1973, pp.~1--11.

\bibitem[WF74]{WF74}
Robert~A. Wagner and Michael~J. Fischer, \emph{The string-to-string correction
  problem}, J. ACM \textbf{21} (1974), no.~1, 168--173.

\end{thebibliography}
	
\end{document}